\documentclass[11pt]{article}
\usepackage[svgnames]{xcolor}

\usepackage{tikz}
\usepackage{tkz-graph}
\usetikzlibrary{positioning}
\usetikzlibrary{shapes}

\usepackage{graphicx}
\usepackage[utf8]{inputenc}
\usepackage[T1]{fontenc}
\usepackage{stmaryrd}
\usepackage{textcomp,amssymb,amsmath,amsthm}
\usepackage{bbm}
\usepackage[english]{babel}
\usepackage[normalem]{ulem}
\usepackage{float} 
\usepackage{tabularx} 

\makeatletter
\newtheorem*{rep@theorem}{\rep@title}
\newcommand{\newreptheorem}[2]{%
\newenvironment{rep#1}[1]{%
 \def\rep@title{#2 \ref{##1}}%
 \begin{rep@theorem}}%
 {\end{rep@theorem}}}
\makeatother
\newreptheorem{theorem}{Theorem}

\makeatletter
\newtheorem*{rep@corollary}{\rep@title}
\newcommand{\newrepcorollary}[2]{%
\newenvironment{rep#1}[1]{%
 \def\rep@title{#2 \ref{##1}}%
 \begin{rep@corollary}}%
 {\end{rep@corollary}}}
\makeatother
\newrepcorollary{corollary}{Corollary}

\usepackage[a4paper,top=3cm,bottom=2cm,left=3cm,right=3cm,marginparwidth=1.75cm]{geometry}

\usepackage{algorithm2e}
\usepackage{amsmath}
\usepackage{amssymb}
\usepackage{graphicx}
\usepackage[colorinlistoftodos]{todonotes}
\usepackage[colorlinks=true, allcolors=blue]{hyperref}

\newtheorem{theorem}{Theorem}

\newtheorem{assumption}{Assumption}
\newtheorem{claim}{Claim}

\newtheorem{lemma}{Lemma}
\newtheorem{corollary}{Corollary}

\newcommand{\util}{\pi}
\newcommand{\Util}{\Pi}
\newcommand{\ppath}{P_{n}}
\newcommand{\throttle}[1]{}

\title{The Empirical Core of the Multicommodity Flow Game Without Side Payments }
\author{Coulter Beeson, Bruce Shepherd}

\begin{document}
\maketitle

\begin{abstract}

Policy makers focus on stable strategies as  the ones adopted by rational players. If there are many stable solutions, however,  an important secondary question is how to select amongst them.  We study this question for the
{\em  multicommodity flow coalition game}, introduced by Papadimitriou  to  model incentives and cooperation between autonomous systems in the Internet. In short, the strategies of the game are flows in a capacitated network (the supply graph). The payoff to any node is the total flow which it terminates.  
Markakis-Saberi  show that this game is balanced and hence has a non-empty core by Scarf's Theorem. In the transferable utility (TU) version this also leads
to a polynomial-time algorithm to find core elements but for the application to autonomous systems,  side payments are not  natural. 
Finding core elements in NTU games, however, tends to be computationally much more difficult, cf. \cite{conitzer}.  Even for this multiflow game, the only previous  result is due to Yamada and Karasawa  who give a procedure to find a core element when the supply graph is a path.   We extend their work by designing an algorithm, called {\sc incorporate}, which  produces  many different core elements. 

We use our algorithm to evaluate several specific instances by running {\sc incorporate} to  generate multiple core vectors. We call these  the {\em empirical core} of the game.
    We find that sampled core vectors are  more consistent with respect to social welfare (SW) than for fairness (minimum payout). For SW they tend to do as well as the optimal linear program value $LP_{sw}$. In contrast, there is a larger range across  fairness in the empirical core;  the fairness values  tend to be   worse than the optimal fairness LP value $LP_{fair}$. We study this discrepancy in the setting of general graphs with single-sink demands.
     In this setting we give an algorithm which produces core vectors that simultaneously maximize SW and fairness. This leads to the following bicriteria result  for general games.
 Given any core-producing algorithm and 
 any $\lambda \in (0,1)$,  one can produce an approximate core vector
 with  fairness (resp. social welfare) at least $\lambda LP_{fair}$ (resp. $(1-\lambda) LP_{sw}$). 

\end{abstract}

\section{Introduction}

\subsection{Coalition Games}

A (non-transferrable) {\em coalition game} consists of a set of players $N$ and a {\em characteristic function} $\Util$
which maps
each $S \subseteq N$ to a subset of $\mathbb{R}^N$. 
The interpretation is that $\Util(S)$ denotes the set of possible payoff (or utility) vectors $\util$ available to players of $S$ if they decide to cooperate
(we assume that $\util_i=0$ if $i \not\in S$).
A general theme in cooperative game theory is to find strategies whereby the {\em grand coalition}, namely $N$ itself,  becomes a stable set of partners.
In other words, we seek a payoff vector $\util \in \Util(N)$ which has no {\em breakaway set}, that is a proper subset $S$ of the players who could do better if they deviate from a grand coalition strategy which produces $\util$.
 We now define this  formally.

 For two vectors $x,y \in\mathbb{R}^N $, we say  $x$ {\em dominates} $y$ {\em on $S$} if $x_i > y_i$ $\forall i \in S$;
 we  write $x \succ_S y$.   
 Let $S$ be a proper subset of $N$ and  $\util \in \Util(N)$.  We call a second payoff vector $\util' \in \Util(S)$ an {\em $S$-deviation from $\util$} if
 $\util' \succ_S \util$; we  also refer to the strategies inducing this payoff as a deviation.
We call $S$ a {\em breakaway set of $\util$} if there is some  {\em $S$-deviation}.

 The {\em core} of a coalition game is the set of vectors $\util \in \Util(N)$ which have no breakaway sets.
 Thus core vectors represent payoffs $\util$ to the grand coalition which are stable in the sense that no subset of players is motivated to defect from the strategies which induce $\util$. 

Core vectors are the idealized  outcomes of rational play in a  coalition game, but what role  can they play in practice? Given that there  may be many payoff vectors in the core, 
{\bf Q1.} {\em which ones are preferable and hence should be incentivized?}
In  a game without side payments  it becomes essential to understand payouts for individual  players.  
We develop theory in order to produce a large number of core vectors for
an (NTU)  multi-flow game. This allows one to sample from the associated {\em empirical core} in order to compare core vectors with respect to social welfare, fairness and other performance metrics.

\subsection{The Multiflow Coalition Game}

The players in a {\em multicommodity flow coalition game} consist of the nodes in a given {\em supply graph} $G=(V,E)$; we refer to $e \in E$ as a  supply edge. In addition, we are given {\em capacities} $c_v$ on each node $v \in V$. 
We are also given a {\em commodity graph} $H=(V,F)$. We  refer to $e \in F$ as a commodity edge; each such edge has an associated  {\em demand} $d_e \geq 0$.

{\em Strategies} in a flow coalition game arise from feasible flows in $G$ for the commodities $H$.
For each commodity $uv \in F$, we denote by $\mathcal{P}_{uv}$ the set of (simple) paths in $G$ joining $u$ and $v$; $\mathcal{P}$ denotes the set of all simple paths.  A {\em flow} is then defined by a non-negative assignment of flow $f_P$ to each path $P$ joining the endpoints of some commodity. Formally, $f$ is a feasible flow if it satisfies:

\begin{enumerate}
\item $\sum_{P \ni v} f_P \le c_v \quad \forall v \in V$ ~~~~ {\em Capacity constraint}

\item $\sum_{P \in \mathcal{P}_{uv}} f_P \le d_{uv} \quad \forall uv \in F=E(H)$~~~~~~~{\em Demand constraint}

\item $f_{P} \ge 0, ~~~P \in \mathcal{P}$.
\end{enumerate}

We let $\mathcal{F}(G,c,d)$ denote the set of all feasible flows. We are also interested in the
strategies available to a subset $S$ of players. We denote by  $\mathcal{F}(G[S],c,d)$, or simply $\mathcal{F}^S$, the feasible flows in
 $G[S]$, the subgraph of $G$ induced by $S$.\\

Each feasible flow induces {\em payoffs} (or {\em utilities}) for the players as follows.
For any $S \subseteq V$, $f \in  \mathcal{F}^S$, and $v \in V$ we define
 $\util_v(f) = \sum_{wv \in E(H)} f_{wv}$ as the sum of all flows that terminate at $v$ (we make no distinction between traffic coming from or going to $v$). Note that if $v \not\in S$, then $\util_v(f)=0$. The set of payoffs  available to a coalition 
 is thus $\Util(S) = \{ \util(f) : f \in \mathcal{F}^S \} $. Given that the set of feasible flows is convex and that the utility function is linear, $\Util(S)$ also defines a convex polyhedron.\\

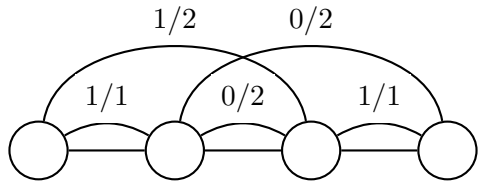
\begin{figure}[ht]
\label{fig:nonconvex}
     \centering
      \begin{tikzpicture}[-,>=stealth',auto,node distance=1.0cm,
                        thick,main node/.style={circle,draw,minimum size=2em,font=\sffamily\Large\bfseries}]
    
        \node[main node] (1) {};
        \node[main node] (2) [right = of 1] {};
        \node[main node] (3) [right = of 2] {};
        \node[main node] (4) [right = of 3] {};

        \Edge(1)(2)
        \Edge(2)(3)
        \Edge(3)(4)

        \tikzset{EdgeStyle/.append style = {bend left = 80}}
        \Edge[label=$1/2$](1)(3)
        \Edge[label=$0/2$](2)(4)
        
        \tikzset{EdgeStyle/.append style = {bend left = 30}}
        \Edge[label=$1/1$](1)(2)
        \Edge[label=$0/2$](2)(3)
        \Edge[label=$1/1$](3)(4)

    \end{tikzpicture}
    \caption{A feasible flow for capacity vector $(\infty,2,2,\infty)$ with a resulting utility vector $(2,1,2,1)$. Case analysis shows that this, and the vector $(1,2,1,2)$ are in the core, but the convex combination   $(\frac{3}{2},\frac{3}{2},\frac{3}{2},\frac{3}{2})$  is not.  $f/d$ indicates the flow $f$ and demand $d$ for a commodity edge.}
\end{figure}

In order to study question {\bf Q1},  one   needs a method to  produce multiple  core vectors. Unfortunately, for  NTU games, there is no general method to (efficiently) compute even one!  A main computational stumbling block is that the core need not be convex (even for multiflow games, see Figure~\ref{fig:nonconvex}) and hence standard optimization techniques are not immediately forthcoming. To date, the only positive result  \cite{yamada2006}  is one that  exhibits a core vector in the case when $G$ is a path. 

We design a simple polynomial-time  algorithm
which produces many core vectors when the supply graph is a path. We then apply  our method to several games and compare the 
payoff vectors
which are generated, called the {\em empirical core}.   In Section~\ref{sec:empirical} we present our findings.  A key takeaway is:  in terms of social welfare, one doesn't go far wrong using any core element. In terms of fairness, however, there can be  a 50\%  gap between the average core vector and the maximum core vector fairness. 

In the next section we outline the technical plan for designing and analyzing our new algorithm.

\subsection{A Method for Computing (Many) Core Vectors}

In \cite{yamada2006} it is proved  that the following simple greedy procedure produces
a vector in the core of the multi-flow game when the supply graph $G$ is a path $123 \ldots n$.  That is, $V(G)=[n]$ and there exist edges $i(i+1)$ for each $i < n$. They process the nodes from ``left to right'' ({\em i.e.}, from smallest to largest). We call this the {\sc YK} algorithm. We denote this graph by $\ppath$.  For each node $i$, they then scan its incident demand edges $ij$ from largest $j$ to smallest. In {\em scanning} a demand $ij$ they {\em route} as much flow as possible on the (unique) path $P_{ij}$ joining $i$ and $j$. They then decrement node capacities accordingly. We refer to this as the core vector resulting from a {\em left-right scan}.\footnote{The article \cite{yamada2006} also refers to an extension to a special class of trees called spiders.} 
As we also run a type of greedy algorithm, we formalize the routing subroutine as follows.  

\vspace*{.3cm}

\begin{algorithm}
\SetAlgoLined
{\sc Route}\\
{\bf Input.} A supply graph $\ppath$ with node capacities $(C_v: v \in P)$, and a subpath $P_{k\ell}$ \\
        \hspace*{.5cm} $~~~~m \leftarrow \min_{v \in P_{k\ell}}\big\{C_v\big\}$ \\
  \hspace*{.5cm} $~~~~f_{k\ell} \leftarrow \min\big\{d_{k\ell},m\big\}$ \\
  \hspace*{.5cm} $~~~~C_v \leftarrow C_v - f_{k\ell} ~~ \forall v \in P_{k\ell} $
\end{algorithm}

\noindent
Here we use shorthand $f_{k\ell}$ to denote the flow variable $f_{P_{k\ell}}$.
Hence we route as much flow for a commodity $k\ell$ as possible, {\em i.e.},  the minimum of the smallest residual capacity on $P_{k\ell}$ and  the demand $d_{k\ell}$.

One property of the order in which  commodities are processed in the {\sc yk} algorithm is that if $[k\ell] \subset [ij]$, then $k\ell$ is processed before  $ij$.  We call this a {\em nested ordering}.  It is not the case that greedily routing flows in a nested order is sufficient to produce a core vector. The first flow in Figure \ref{fig:nest-counter} is produced by a nested ordering, and the second shows a breakaway set (the middle $4$ nodes).

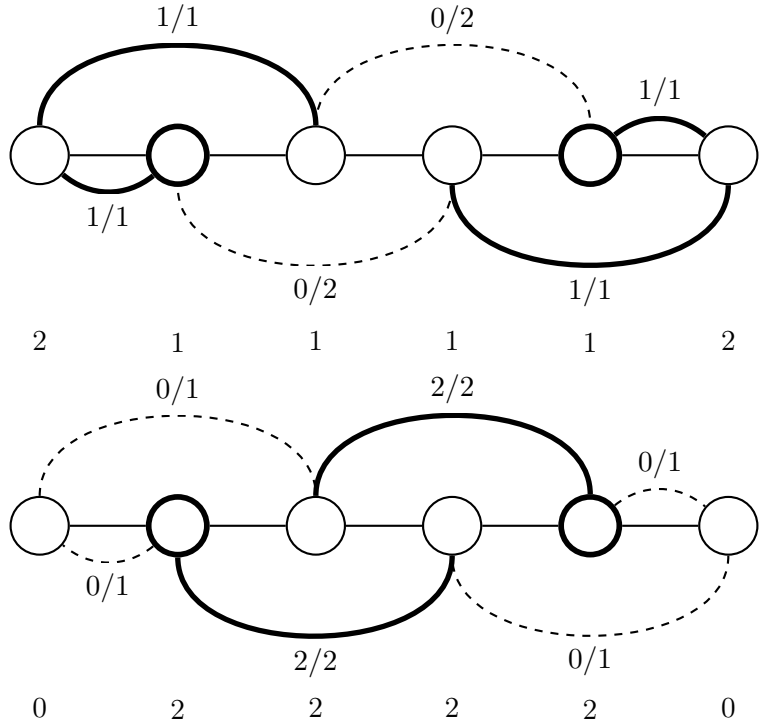
\begin{figure}[htbp]
     \centering
      \begin{tikzpicture}[-,>=stealth',auto,node distance=1.0cm,
                        thick,main node/.style={circle,draw,minimum size=2em,font=\sffamily\Large\bfseries}]
    
        \node[main node, label={[label distance=1.8cm]270:2}] (1) {};
        \node[main node, line width=0.75mm, label={[label distance=1.8cm]270:1}] (2) [right = of 1] {};
        \node[main node, label={[label distance=1.8cm]270:1}] (3) [right = of 2] {};
        \node[main node, label={[label distance=1.8cm]270:1}] (4) [right = of 3] {};
        \node[main node, line width=0.75mm, label={[label distance=1.8cm]270:1}] (5) [right = of 4] {};
        \node[main node, label={[label distance=1.8cm]270:2}] (6) [right = of 5] {};

        \Edge(1)(2)
        \Edge(2)(3)
        \Edge(3)(4)
        \Edge(4)(5)
        \Edge(5)(6)

        \tikzset{EdgeStyle/.append style = {bend left = 40}}
        \Edge[label=$1/1$,lw=2pt](2)(1)
        \Edge[label=$1/1$,lw=2pt](5)(6)

        \tikzset{EdgeStyle/.append style = {bend left = 90}}
        \Edge[label=$1/1$,lw=2pt](1)(3)
        \Edge[label=$1/1$,lw=2pt](6)(4)
        
        \Edge[label=$0/2$, style=dashed](4)(2)
        \Edge[label=$0/2$, style=dashed](3)(5)
        
    \end{tikzpicture}

      \begin{tikzpicture}[-,>=stealth',auto,node distance=1.0cm,
                        thick,main node/.style={circle,draw,minimum size=2em,font=\sffamily\Large\bfseries},x=0.35cm,y=0.35cm]
    
        \node[main node, label={[label distance=1.75cm]270:0}] (1) {};
        \node[main node, line width=0.75mm, label={[label distance=1.75cm]270:2}] (2) [right = of 1] {};
        \node[main node, label={[label distance=1.75cm]270:2}] (3) [right = of 2] {};
        \node[main node, label={[label distance=1.75cm]270:2}] (4) [right = of 3] {};
        \node[main node, line width=0.75mm, label={[label distance=1.75cm]270:2}] (5) [right = of 4] {};
        \node[main node, label={[label distance=1.75cm]270:0}] (6) [right = of 5] {};

        \Edge(1)(2)
        \Edge(2)(3)
        \Edge(3)(4)
        \Edge(4)(5)
        \Edge(5)(6)

        \tikzset{EdgeStyle/.append style = {bend left = 40}}
        \Edge[label=$0/1$, style=dashed](2)(1)
        \Edge[label=$0/1$, style=dashed](5)(6)

        \tikzset{EdgeStyle/.append style = {bend left = 90}}
        \Edge[label=$0/1$, style=dashed](1)(3)
        \Edge[label=$0/1$, style=dashed](6)(4)
        
        \Edge[label=$2/2$,lw=2pt](4)(2)
        \Edge[label=$2/2$,lw=2pt](3)(5)

    \end{tikzpicture}
    \caption{Greedy nested ordering is not sufficient to create core vectors.}
    \label{fig:nest-counter}
\end{figure}

In some sense, the issue in this example is that we grow two ``islands'' of flows (the left and right ends of the paths) but the internal section formed a breakaway.
To address this issue we seek  nested orderings which are gown in a contiguous fashion. We build a flow iteratively starting from any node
and then  {\em incorporating}  nodes one by one, enforcing  the incorporated nodes to induce a connected subgraph.   We now present our algorithm
for an arbitrary connected supply graph $G$.

\vspace*{.5cm}

\begin{algorithm}[H]
\SetAlgoLined

{\sc incorporate}

{\bf Input.} Given a Game (G,c,H,d), starting node $r$\\
{\bf Output.} a flow $f$ such that $\util(f) \in \textit{Core}(G,c,H,d)$\\

$f \leftarrow 0$, $~~S \leftarrow \{r\}$, $~~(C_v \leftarrow c_v: v \in V(G))$\\
\While{$S \subset V$}{
        Let $v \in N(S)$\\
        Let T be a stack with elements of S in increasing  order by distance from v\\
        $S \leftarrow S+v$\\
        \While{$T \neq \emptyset$}{
            $k \leftarrow Pop(T)$\\
            \If{$kv \in E(H[S])$}{
                \bf{route}$\big(P_{kv}\big)$\\
            }
        }
}
Return $f$
\end{algorithm}
\vspace*{.5cm}

When a player $v$ is incorporated we route all the commodities incident to that player whose other end point $u$ is already incorporated. We do this to respect a nested ordering, {\em i.e.}, we process nodes $u$ closer to $v$ first. The terminology is especially fitting as both "incorporate" and "core" etymologically derive from the Latin {\em corpus}. 
The proof of our main theorems shows that at any point in the execution of the algorithm,  the current flow is in the core of the sub-game for the incorporated set.

\begin{theorem}
\label{thm:main}
If $G$ is a path, then {\sc incorporate}  returns a flow whose payoff vector is in the core. 
\end{theorem}

We prove Theorem~\ref{thm:main} in two main steps.  In Section~\ref{sec:certificate} we give a class of {\em certificates}. For a given flow $f$ and proper subset $S \subseteq V$
either $S$ has a deviation, or one of our certificates  guarantees that $S$ is not a breakaway set for the payoffs of $f$. These arguments are based on duality for LPs with strict inequalities.
The results  apply to general supply graphs.

In Section~\ref{sec:pathcert} we specialize the structure of these certificates to supply graphs which are paths.  We then use this structure to show that for any output of {\sc incorporate}  and any $S$, we can produce such a certificate.

In \cite{yamada2006} it is reported that their approach extends to supply graphs which are spiders, although an argument is not provided.
 In Section~\ref{sec:pathcert} we  show how to use  {\sc incorporate} to
 derive an algorithm for finding core vectors in spiders. 

\begin{corollary}
\label{cor:spiders}
Let $\mathcal{T}_c$ be the class of trees (so-called spiders) with at most $1$ node of degree greater than two.
There is a polytime algorithm for producing core elements for any  multi-flow game associated with a tree in $\mathcal{T}_c$. 
\end{corollary}

\subsection{The Empirical Core and Improving Fairness }
\label{sec:emp}

In Section~\ref{sec:empirical} we leverage our main algorithmic result.  We  generate an array of games based on different demand models.
For each game,  we can run {\sc  incorporate} from any starting node and  incorporate nodes in any (valid) order.
This produces a large number of distinct core elements. We call this set of samples the {\em empirical core} and we explore their properties. Two of the most prominent features are:
\begin{itemize}
    \item {\bf Observation 1.} {\em Sampled core payoffs have social welfare which is very close to the theoretical optimum,}
      \item {\bf Observation 2.} {\em There is a  wide range on fairness (minimum payoff) amongst the sampled core payoffs,} 
\end{itemize}

\noindent
thus suggesting that incentives are required in order to achieve both the maximum social welfare and fairness.

There is one setting where it is possible
to balance the competing objectives of social welfare and fairness exactly.  That is, we can find a core solution which achieves both the optimal fairness and social welfare. 
Instead of restricting the supply graph (to trees or spiders), we   allow general supply graphs but  restrict the topology of the commodity graph.
In this vein, it becomes illuminating to consider the class of  {\em single-sink}  commodity graphs, where all commodities are incident to a ``root'' node $t$.
We can associate each such  commodity with {\em terminals} $s_i \in V(G)$ and demand $d_i$, $i=1,2,3 \ldots k$.  A payoff
vector can be viewed as $x \in \mathbb{R}^{[k]}_{\geq 0}$ for which there is a feasible flow that  routes $x_i$ flow from each $s_i$ to $t$;
the payoff to each $i$ is $x_i$ and to $t$ it is $\sum_i x_i$. It  easily follows that any $x$ which maximizes $\sum_i x_i$ is a core element. This is because any breakaway set must contain $t$, and it already achieves its maximum utility (globally content in our parlance).

We use the standard LPs to study the quality of our core vectors. For instance, the {\em fairness LP} for a multiflow game consists of the standard formulation together with a parameter $\tau \geq 0$ with the following constraints
\[
\sum_{j} f_{ij} \geq \tau ~~~~~\mbox{for each $i$}.
\]

\noindent
If this LP is feasible we say the instance has {\em fairness} of at least $\tau$. We may similarly define the social welfare LP.

In Section~\ref{sec:fairness} we show the following.

\begin{theorem}
\label{thm:SS}
Let $\mathcal{F}$ be the family of multiflow instances $(G,H)$ where
$G$ is any capacitated supply graph  and $H$ is a single-sink commodity graph. Then there is a polytime algorithm which given an instance in $\mathcal{F}$, produces a core vector which simultaneously achieves the optimal LP fairness and LP  social welfare.
\end{theorem}

\noindent
This exact result is in contrast to the single-sink setting where agents pay for edges used to carry traffic for single-sink flows (that is minimum spanning tree cost sharing games)  \cite{gupta2004cost}.

 Unfortunately, for general demands  - and even for line graphs - we are doomed to fail with this approach.  Consider $P_4$ with unit capacity nodes. We also have two unit demand commodities: $14,23$. One easily deduces that the only core vector is obtained by the flow which routes 1 unit of demand for $23$. Hence, the fairness is $0$ for any core vector but the fairness LP achieves $1/2$.

One can however balance these objectives if  approximate core solutions are allowed -- see Section~\ref{sec:notation}.

\begin{theorem}
\label{thm:bicriteria}
Consider a family of multiflow games for which we have a core-producing algorithm.
Consider an instance  for which the fairness is at least $\tau$.
For any $\lambda \in (0,1)$ we can compute a $\frac{1}{1-\lambda}$-approximate core vector with fairness at least $\lambda \tau$.
\end{theorem}

We state this here in terms of the multiflow game, but the theorem holds much more generally. We prove the general case in Section~\ref{sec:fairness}. 


\subsection{Related Work}

The Multicommodity Flow Game was originally introduced by Papadimitriou to model incentives in internet routing \cite{Papadimitriou2001}. In this model each node or player represents an autonomous system (AS). Each AS is a component of the internet administrated  (usually) by a single  operator, such as a university or corporation. 
The assumption is that each AS  desires to  route as much of traffic originating from its users as possible. On the other hand ASs must cooperate with each other to allow their customers access to the broader network.

One difficulty is that the core may not a exist. Scarf however offers a sufficient condition for the core to be non-empty. To introduce this we first define {\em balanced collections}. A collection of coalitions $T \subseteq \mathcal{P}(V)$ is {\em balanced} if there exists weights $w_S$ such that for every player $i$, $\sum_{S \ni i} w_S=1$. If all weights are in $\{0,1\}$, then this is exactly a partition of $V$. We say that a payoff $\util$ is {\em attainable} by $S$ if $\util \in \Util(S)$. A game is {\em balanced} if for all balanced collections $T$, if $\util^S$ is attainable for all $S \in T$, then $\util$ is attainable for the grand coalition. Scarf's theorem  says that every balanced game has a non-empty core \cite{scarf1967core}.

Markakis and Saberi \cite{Markakis2005} show that the (NTU or TU) Multicommodity Flow Game is balanced and thus has a non-empty core. In fact as every sub-game of a multiflow game is still another multi-flow game, this result shows that these games are {\em totally balanced}. This leads to a method for finding a core vector in the TU case only (from a construction based on LP duality). For the NTU multicommodity flow game, the only efficient algorithm for computing a core vector is the algorithm due to  Yamada and Karasawa \cite{yamada2006}, denoted earlier by {\sc yk}.

\subsection{Notation and Basics}
\label{sec:notation}

We say  a  payoff vector $\util \in \Util(N)$ is
in the {\em $(1+\epsilon)$-core} of a game, if there is no
proper subset $S$ and $\util' \in \Util(S)$ such that
$\util' \succ_S (1+\epsilon) \util$.
We refer to a game as {\em downwards-closed} if for each $S$, $\Util(S)$
is a downwards-closed polytope in $\mathbb{R}^N_{\geq 0}$. The following is straightforward.  
\begin{lemma}
\label{lem:scale}
Let $\Gamma$ be a downwards-closed game and $\Gamma'$ be the game where
$\Util'(S)=\{\frac{1}{1+\epsilon} \util:  \util \in \Util(S)\}$. Then any core vector in $\Gamma'$ is a $(1+\epsilon)$-approximate core vector for $\Gamma$.
\end{lemma}

For most of the technical parts we work with a feasible flow $f$ with induced payoff vector $\util$. 
We call a node {\em tight} if all of its capacity is used up by the  flow $f$.
We say  commodity is {\em fully-routed} if the demand constraint is satisfied with equality.
If $f_P>0$ we refer to $P$ as a {\em positive flow path}. We say that a path $P$ {\em touches} the node $v$ if $v \in V(P)$. We say $P$ {\em transits} a node $v$ if $v \in I(P)$, where $I(P)$ denotes the internal nodes of $P$. A commodity $k\ell$ is called {\em positive} if
some flow is routed between $k$ and $\ell$. 

Recall that the flows from {\sc incorporate} follow a {\em nested ordering}.
Formally, this means that if $P$ is a subpath of $Q$, then the algorithm
would attempt to route on $P$ before routing on $Q$ (assuming both paths are associated with a commodity). 
It is helpful to summarize  properties of 
such nested flows in path supply graphs.
\begin{lemma}
\label{lem:nested-down}
If $f_{k\ell} > 0$, then $f_{ij}=d_{ij}$, for all $[i,j] \subset [k,\ell] $
\end{lemma}

\begin{proof}
If $f_{k\ell}>0$, then prior to this commodity being routed, $C_v>0 \quad  \forall v \in [k,\ell]$. Because all commodities $ij$ such that $[i,j] \subset [k,\ell]$ are processed prior to $f_{k\ell}$, we must have had $C_v>0$ for every $v \in [i,j]$ when $ij$ is routed. 
\end{proof}

\begin{lemma}
\label{lem:nested-up}
If $f_{k\ell} < d_{k\ell}$, then $f_{ij}=0$,  for all  $[i,j] \supset  [k,\ell]$
\end{lemma}

\begin{proof}
If $f_{k\ell} < d_{k\ell}$, then there is a tight node $x \in [k,\ell]$
after $k\ell$ is processed. Since any  commodity $[i,j] \supset [k,\ell]$ is processed later,  $x$ will block $ij$ from routing any flow.
\end{proof}

\section{Certifying that $S$ is not a Breakaway Set in General Graphs} 
\label{sec:certificate}

Let $\hat{f}$ be some feasible flow with payoff $\util$. How can we check if $\util$ is  in the core? Let us consider a more restricted question. For  a
given set $S$,  can we verify that it has no deviation? That is,
can we certify that there is no flow $f \in \mathcal{F}^S$ such that $\util^S \big(f\big) > \util^S \big(\hat{f} \big)$?  The existence of such a deviation corresponds to   feasibility of the following strict-inequality linear program.

\begin{equation}
\label{eqn:LP}
\begin{array}{ll@{}ll}
\text{maximize}  & \displaystyle 0^Tf &\\
\text{subject to}& \displaystyle\sum\limits_{P \ni v}   &f_P \leq  c_v     &v \in S \\
                 & \displaystyle\sum\limits_{P \in \mathcal{P}_{k\ell}[S]}  &f_P \le d_{k\ell} &k\ell \in H[S] \\
                 & &\util_v \big(f\big) > \util_v \big(\hat{f}\big) & v \in S \\
                 & &f_P \ge 0, ~~P \in \mathcal{P}[S]
\end{array}
\end{equation}

\noindent
Here we use $H[S]$ to denote the subgraph of $H$ induced by $S$
and $\mathcal{P}[S], \mathcal{P}_{k\ell}[S]$ denote simple paths in $G[S]$. 

To show that $S$ does not have a deviation $f$, we must show that the above linear program is infeasible.
By Farkas' Lemma, a linear system $Ax \le b, x \ge 0$ is infeasible if and only if $\exists y \ge 0$ such that $y^TA \geq 0$ and $y^Tb<0$. This would immediately give a way to show that a coalition $S$ is not a breakaway, but for the fact that (\ref{eqn:LP}) has  strict inequalities. Hence, for some $\epsilon > 0$, consider a standard LP, call it $LP(\epsilon)$, obtained by modifying the {\em payoff constraints} to be: $-\util_x\big(f\big) \leq - \util_x\big(\hat{f}\big) - \epsilon$. We may now apply Farkas' Lemma to this new program to  see that it is infeasible if and only if there is a solution $(y_x: x \in S),(z_{k\ell}: k\ell \in E(H)),(w_x: x \in S)$ to the following system.

\begin{eqnarray}\label{eqn:linecertificate}
\label{eqn:total-eps}
   \displaystyle \sum_{v \in S} y_v c_v + \sum_{k\ell 
 \in H[S]} z_{k\ell} d_{k\ell} &<& \sum_{v \in S} w_v (\util_v \big(\hat{f} \big)+\epsilon)   \\
 \label{eqn:demands}
 \sum_{v \in V(P)} y_v + z_{k\ell} &\ge& w_k + w_{\ell} \quad    ~~~~\forall  \mbox{$P \in \mathcal{P}_{k\ell}[S], ~k\ell \in H[S]$ }  \\
 \label{eqn:noneg}
 y,z,w  &\geq& 0
\end{eqnarray}

\noindent where we use that if $P \in \mathcal{P}_{k\ell}$, then $f_P$ only contributes to the payoffs $\util_k$ and $\util_{\ell}$.

These facts are used  to prove the following.

\begin{lemma}
\label{lem:certificate}
$S$ does not have a deviation if and only if there exists $y,z,w \geq 0$ which satisfy $(3,4)$ and
\begin{equation}
\label{eqn:total}
 \quad \displaystyle \sum_{v \in S} y_v c_v + \sum_{k\ell 
 \in E(H[S])} z_{k\ell} d_{k\ell}  \leq \sum_{v \in S} w_v \util_v \big(\hat{f}\big) 
\end{equation}
 
\noindent
and  such that $w \neq 0$.
\end{lemma}
\begin{proof}
Note that $S$ does not have a deviation if and only if  $LP(\epsilon)$ is infeasible for any $\epsilon > 0$. First suppose that
 $y,z,w$  satisfy the prescribed conditions. Since $w_x > 0$ for some $x$, we then
 have that the strict inequality holds (\ref{eqn:total-eps}), for any choice of $\epsilon > 0$. That is, $LP(\epsilon)$ is infeasible for any $\epsilon > 0$. 
 Conversely, if $S$ does not have a deviation, then $LP(1)$ is infeasible and hence there is an associated dual certificate $y,z,w$. Since the left-hand side of (\ref{eqn:total-eps}) is non-negative, and $\util\big(\hat{f}\big) \geq 0$, we must have that $w_x >0$ for some $x \in S$.
\end{proof}

We call a feasible solution $y,z,w \neq 0$ to (\ref{eqn:demands}-\ref{eqn:total})  a (Farkas) {\em certificate} for a coalition $S$ as it certifies that $S$ is not a breakaway set. We call such a coalition {\em certifiable}. 
To establish that some vector is in the core, we must argue that any proper subset $S$ is certifiable. Later, we also need the following fact.

\begin{lemma}
\label{lem:up-closed}
If a coalition $S$ is certifiable for a given payoff $\util\big(\hat{f}\big)$, then $S$ is certifiable for any payoff arising from a flow $f \in \mathcal{F}^V$ such that $f \ge \hat{f}$.
\end{lemma}
\begin{proof}

Clearly payoffs are monotonic with regard to flow, and the payoffs only occur in constraint (\ref{eqn:total}). As increasing the payoff can only improve this constraint, $S$ is still certifiable under $f$. 
\end{proof}

\subsection{Implied Values and Trivial Cases}
\label{sec:implied}

We continue to work on the case for general supply graphs $G$.
For the remainder of this section we assume that we are given a payoff vector $\util$ (and a flow $f$ which induces it) and a proper subset $S$ and our goal is to certify that $S$ is not a breakaway set. We first restrict attention to certificates with a simpler structure.  We then consider several cases where it is easy to certify $S$.

Given  vectors $y,w \geq 0$ we can easily determine if there is a $z$ for which $y,z,w$ is a certificate. This is because 
 the ``best'' choice of $z$ is for each $k\ell \in H[S]$, to set
 
 \begin{equation}
 \label{eqn:implied}
 z_{k\ell}:=Z_{k\ell}(y,w)=\max_{P \in \mathcal{P}_{k\ell}}\Big\{w_k+w_{\ell}-\sum_{v\in V(P)} y_v,0 \Big\}.
 \end{equation}
 
We refer to $Z_{k\ell}(y,w)$ as the value {\em implied} by $y,w$. This is the best value since  $z_{k\ell}$ only occurs in 2 constraints: in (\ref{eqn:demands}) where it needs to be at least the implied value, and (\ref{eqn:total}) where making it smaller improves the inequality.

\begin{lemma}
\label{lem:standard-z}
If $y,z,w$ is a certificate, then so is $y,z',w$  where $z'$ is set to the values implied by $y,w$.
\end{lemma}

\noindent
From  now  on, we assume  that $z$ is set to the implied values. Hence we  only need to check condition (\ref{eqn:total}). 

We will only consider {\em binary  certificates}, that is,
 the vectors $y,w$ are binary.  We always work with certificates where  the support of $y$, denoted by $Y$, are tight nodes. Moreover, we always select certificates so that $Y \subseteq W$, where $W$ is the support of $w$. 
 
 We now catalogue a few cases where certificates  are trivial to construct.

We call a node $v$ {\em globally content} if $\util_v=c_v$. We hardly need a certificate to show that no feasible flow (for any coalition $S$!) could yield strictly greater utility for $v$.
Nevertheless we construct a certificate as a warm-up for later results.

\begin{lemma}
If there exists a globally content player $v \in S$, then $S$ is certifiable.
\end{lemma}
\begin{proof}
We construct the certificate as $\{v\} = Y = W$. Note that the right hand side of (\ref{eqn:demands}) is only non-zero for demands
which terminate at $v$. For such demands $y_v=1$ already handles this constraint. Then the implied values are $z=0$.\\

\noindent
Now for condition (\ref{eqn:total}) we have 

\begin{align*}
    yc + zd &= c_v \\
            &= \util_v \\
            &= w\util
\end{align*}

\end{proof}

\noindent
We now assume there are no globally content players and hence:
\begin{assumption}
\label{ass:global-content}
Every tight player in $S$ transits some flow. 
\end{assumption}

If there exists $v \in S$ such that $\pi_v \ge \sum_{u \in S} d_{uv}$, then no feasible flow from $\mathcal{F}^S$ can route more flow to $v$. We say that such a player $v$ is {\em $S$-content}.\\

\begin{lemma}
\label{lem:s-cont}
If there exists an {\em $S$-content} player, then $S$ is certifiable
\end{lemma}
\begin{proof}
We set $w_v=1$ and $y=0$. Hence the implied values are $z_{kv}=1$ for all commodities $kv \in H[S]$.\\

\noindent
Again we check condition (\ref{eqn:total}).

\begin{align*}
    yc + zd &= zd \\
            &= \sum_{kv \in H[S]} d_{kv} \\
            &= \sum_{kv \in H[S]} f_{kv} \\
            &\le \sum_{k \neq v} f_{kv} \\
            &= \util_v \\
            &= w\util
\end{align*}

Where the second equality follows from $z$ only selecting commodities incident $v$, and the third equality follows as every commodity (from $H[S]$) incident $v$ is fully-routed ({\em i.e.} $d_{kv}=f_{kv}$).

\end{proof}

\noindent
Henceforth we assume the following. 
\begin{assumption}
\label{ass:s-content}
Every player in $S$ must be incident to a commodity within $H[S]$ that is not fully-routed.  
\end{assumption}

Note that if there are no tight nodes in $S$, then every player in $S$ is $S$-content. Hence we assume the following.

\begin{assumption}
\label{ass:tight}
The coalition $S$ contains a tight node.
\end{assumption}

\subsection{Pre-Certificates}

We continue to work in general supply graphs. In this section we prove a key lemma for building
 certificates  based on two sets of nodes $Y \subseteq W \subseteq S$, where $Y$ is a set of tight nodes. 
We call this pair a {\em pre-certificate} if the vectors $y=\mathbbm{1}_Y,w=\mathbbm{1}_W$ can be extended to a certificate by setting a vector $z$ to the implied values (\ref{eqn:implied}).

We say   the variables $y,z,w$  {\em select} some node capacities, demands, and node utilities respectively. We then think of the selected demands and capacities as costs that need to be charged against the selected utilities, {\em i.e.}, to verify the  condition (\ref{eqn:total}) that $w\util- (yc+zd) \ge 0$. In other words, a {\em charging scheme} consists of paying for the selected capacities $yc$ and demands $zd$ by using the payoffs from the selected players $w \util$. \\  

We  now develop   sufficient conditions for sets $Y,W$ to be a pre-certificate.

\begin{lemma}[Pre-Certificate Lemma]
\label{lem:pseudo1}
Consider a pair of node sets $Y \subseteq W \subseteq S$ where nodes in $Y$ are tight.  This forms a pre-certificate if the following properties hold:

\begin{enumerate}
    \item[\textbf{P1)}] every commodity with both endpoints in $W$ is fully-routed. 
    \item[\textbf{P2)}] every positive flow path that transits a player $v \in Y$ has an endpoint in $W$
    \item[\textbf{P3)}] if $ab$ is a  non-fully-routed commodity with an endpoint in $W$, then $G \setminus Y$ does not contain an $ab$-path 
    \item[\textbf{P4)}] no positive flow path touches more than one node in $Y$.
\end{enumerate} 
\end{lemma}

\begin{proof}{\parindent0pt
We let 
$y=\mathbbm{1}_Y,w=\mathbbm{1}_W$.
For each $k\ell \in H[S]$, we define $W(k,\ell) := w_k+w_\ell$ and so 
the implied values for $z$  (\ref{eqn:implied})
are given by $z_{k\ell} = \max_{P \in \mathcal{P}_{k\ell}}\Big\{W(k,\ell)-y(P),0 \Big\}$;
where $y(P)=\sum_{v \in V(P)} y_v$. It follows that  $z_{k \ell} \in \{0,1,2\}$. We claim that if $z_{k\ell}>0$, then $k\ell$ is fully-routed. Since if $W(k,\ell)=2$,
this follows from   \textbf{P1}. Otherwise  $W(k,\ell)=1$ and $y(P)=0$ for some $k\ell$ path, but this violates 
\textbf{P3}, a contradiction.   
\\

It now remains to verify (\ref{eqn:total}). We do this by a charging argument. First, we charge demands selected by $z$ to utilities of certain chosen players in $W \setminus Y$. \\

Case 1) $~z_{k\ell}=2$:
\begin{quote}
Then $W(k,\ell)=2$, and in fact some  $k\ell$-path does not  touch any node of $Y$. By \textbf{P1}, $k\ell$ is fully-routed and we charge this "twice selected" commodity against both endpoints, each endpoint being charged $d_{k\ell}$.
\end{quote}

Case 2) $~z_{k\ell}=1$:

\begin{quote}
There are two subcases, either $W(k,\ell) = 1$ or $W(k,\ell)=2$.

    Case 2.1) $~z_{k\ell}=1$ and $W(k,\ell)=1$:
    \begin{quote}
    We must have that some $k\ell$-path does not touch any node in $Y$. \textbf{P3} thus ensures that any such commodity is fully-routed, and we charge the $z_{k\ell} d_{k \ell}$
    to the unique endpoint in $W$. 
    \end{quote}
    
    Case 2.2) $~z_{k\ell}=1$ and $W(k,\ell)=2$:
    \begin{quote}
        
    By definition of $z_{k\ell}$, there is some  $k\ell$-path which  touches precisely one node in $Y$. If that node is in $\{k,\ell\}$,  then we charge the demand $z_{k\ell} d_{k\ell}$ to the utility of its endpoint {\em not} in $Y$. Otherwise $k\ell$ transits a node in $Y$ (and  \textbf{P1} ensures $k\ell$ is fully-routed). In this case we arbitrarily charge one endpoint $k$ or $\ell$.
    
    \end{quote}

\end{quote}

We have now charged off $zd$ and  turn to the matter of charging the quantity $yc$ to the
remaining (uncharged) utility  from $w \util$. To this end,    notice that the demands considered in Case 2.2) contribute utility to two nodes in $W$ but we have charged its quantity $z_{k\ell}d_{k\ell}$ to one of them. Hence we can still charge to the other endpoint in $W$. A similar situation occurs for fully-routed commodities $k\ell$ for which   $z_{k\ell}=0$ if neither endpoint is in $Y$.  Then $W(k,\ell)=1$ since if $W(k,\ell)=2$,  property \textbf{P4} ensures that $z_{k\ell} >0$.  Hence we did not yet charge the utility from the selected endpoint in $W$. We denote the total such {\em surplus} utility from these two situations by $s$. \\

 Let's classify the uncharged utilities. First, we have only charged players in $W\setminus Y$ so $\util_v$ is available from each $v \in Y$. Second, since $z_{k\ell}>0$ only for fully-routed demands, we have not yet charged any utility accrued from non-fully-routed commodities. Third, we have not charged the utility generated from external demands (that is demands with one end point not in $S$) on (up to) one endpoint in $W$. Finally, as just noted  we sometimes have surplus utility $s$ associated with one of the endpoints of a fully-routed demand.
For each of these three types of excess utility we denote the total available as follows: $s$ (for fully-routed commodities as defined above), $\sigma$ (from positive but non-fully-routed commodities),  $e$ (from external commodities) and $\sum_{v \in Y} \util_v$ (from players in $Y$). Thus the total utility available to charge off $yc$ is
\begin{equation}
\label{eqn:avail}
     w \util - zd \ge s + \sigma + e + \sum_{v \in Y} \util_v.
\end{equation}
\\

We must now charge the quantity $yc$ to these surplus utilities. Since any positive commodity touches  at most one node of $Y$ (by \textbf{P4}) we consider  each node  $v \in Y$ separately when considering how to charge  for its quantity $y_vc_v=c_v$.   This quantity arises  in two ways:  from flow paths that terminate at $v$ or  flow paths that transit $v$. We charge the capacity  used by flows terminating at $v$ to $\util_v$ itself.  For the remaining capacity we have by \textbf{P2} that any such transiting flow path has at least one endpoint in $W$. First, if this is a non-fully-routed commodity, we charge its flow to the unique endpoint in $W$; hence we are charging utility which contributed to $\sigma$. Otherwise, it is a transiting flow path for a fully-routed commodity $k\ell$.   Either $z_{k\ell}=0$ or $W(k,\ell)=2$ since commodities from Case 2.1) do not transit (or even touch) any nodes in $Y$. In both cases this capacity contributed to the surplus $s$ of utility from fully-routed commodities. Hence, there is an  endpoint in $W$ which has available utility to charge.\\

Combining these observations we deduce that 

$$yc \leq s + \sigma +  \sum_{v \in Y} \util_v.$$

\noindent
This together with (\ref{eqn:avail}) yields 
\[
w \util - zd -yc \geq e \geq 0
\]

\noindent
 and hence the implied certificate $y,z,w$ satisfies (\ref{eqn:total}).\\
}\end{proof}

\section{Path Certificates and Proof of Main Theorem}
\label{sec:pathcert}

In the remainder of this section, we assume our supply graph is a path $P_n$. We also have that 
 $f$ is some output from the algorithm {\sc incorporate}.  We show how to certify a given breakaway set $S$. That is, we prove that there exists $y,z,w$ which satisfy (\ref{eqn:demands},\ref{eqn:noneg},\ref{eqn:total}). 
 Without loss of generality $S$ is of the form  $S=[i,j]$ for some $i < j$.  From now on we also make the assumptions (\ref{ass:global-content}-\ref{ass:tight}) from Section~\ref{sec:implied}.


We first eliminate the possibility of positive flow paths which span the interval $[i,j]$.
 
\begin{lemma}
\label{cor:spanpositive}
If there is a positive commodity $k\ell$ such that $S = [i,j] \subseteq [k,\ell]$, then $S$ is certifiable.
\end{lemma}
\begin{proof}
If such a commodity exists and $k\ell \neq ij$, then by Lemma~\ref{lem:nested-down}, every commodity in $H[S]$ is fully-routed. But then  every player in $S$ is 
$S$-content, contradicting  Assumption~\ref{ass:s-content}. If $k\ell=ij$, then every node in $(i,j)$ would be $S$-content as well.
Since there is no $S$-content node, this only leaves the case $j=i+1$. We must have $i(i+1)$ is not fully-routed  
or else both $i$ and $j$ would be $S$-content. But then since $f$ was produced via a nested order,  $i(i+1)$ was routed and blocked before any demands which transit through either $i$ or $i+1$.
It follows that one of $i$ or $i+1$ hit capacity only by commodities which terminate at it. Hence it is globally content, contradicting  Assumption~\ref{ass:global-content}.
\end{proof}

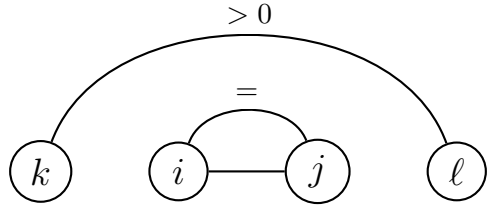
\begin{figure}[ht]
     \centering
      \begin{tikzpicture}[-,>=stealth',auto,node distance=1.0cm,
                        thick,main node/.style={circle,draw,minimum size=2em,font=\sffamily\Large\bfseries}]
    
        \node[main node] (1) {$k$};
        \node[main node] (2) [right = of 1] {$i$};
        \node[main node] (3) [right = of 2] {$j$};
        \node[main node] (4) [right = of 3] {$\ell$};
        
        \Edge(2)(3)
        
        \tikzset{EdgeStyle/.append style = {bend left = 70}}
        \Edge[label={$>0$}](1)(4)
        \Edge[label={$=$}](2)(3)

    \end{tikzpicture}
    \caption{Configuration of Lemma \ref{cor:spanpositive}.}
\end{figure}

\subsection{Anchor Sets for Paths}

We define the {\em anchors} of a node $v$ as $A(v) = \{ k \in V : \exists P \in \mathcal{P}_{k\ell}, ~ f(P)>0, ~ v \in I(P) \}$. In other words, if there is a positive flow path which transits $v$, then the endpoints of the path are anchors of $v$. Given a coalition $S$ we refer to any positive commodity with exactly one end point in $S$ as an {\em external commodity}.

We further refine our notion of anchors for paths. For a given node $v$ we split its anchors into disjoints sets of {\em left} and {\em right anchors}. The left anchors are $L(v)= \{ u \in V : u \in A(v), ~ u<v \}$ and right anchors $R(v)= \{ u \in V : u \in A(v), ~ v < u \}$. We define the left (resp. right) {\em anchor sets} for a player $v$ as $Y=\{v\}$, and $W=L(v) \cup \{v\}$. (resp. $Y=\{v\}$, and $W=R(v)\cup \{v\}$). 

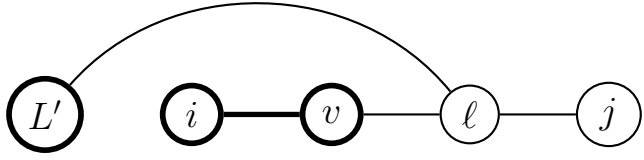
\begin{figure}[ht]
\label{fig:PA}
     \centering
      \begin{tikzpicture}[-,>=stealth',auto,node distance=1.0cm,
                        thick,main node/.style={circle,draw,minimum size=2em,font=\sffamily\Large\bfseries}]
    
        \node[main node, line width=0.75mm] (1) {$L'$};
        \node[main node, line width=0.75mm] (2) [right = of 1] {$i$};
        \node[main node, line width=0.75mm] (3) [right = of 2] {$v$};
        \node[main node] (4) [right = of 3] {$\ell$};
        \node[main node] (5) [right = of 4] {$j$};
        
        \Edge[lw=2pt](2)(3)
        \Edge(3)(4)
        \Edge(4)(5)

        \tikzset{EdgeStyle/.append style = {bend left = 50}}
        \Edge(1)(4)

    \end{tikzpicture}
    \caption{The configuration disallowed by \textbf{PA}}
\end{figure}

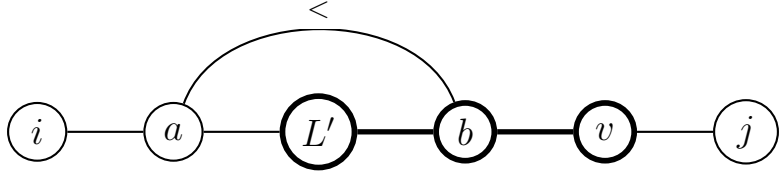
\begin{figure}[ht]
\label{fig:PB}
     \centering
      \begin{tikzpicture}[-,>=stealth',auto,node distance=1.0cm,
                        thick,main node/.style={circle,draw,minimum size=2em,font=\sffamily\Large\bfseries}]
    
        \node[main node] (1) {$i$};
        \node[main node] (2) [right = of 1] {$a$};
        \node[main node, line width=0.75mm] (3) [right = of 2] {$L'$};
        \node[main node, line width=0.75mm] (4) [right = of 3] {$b$};
        \node[main node, line width=0.75mm] (5) [right = of 4] {$v$};
        \node[main node] (6) [right = of 5] {$j$};
        
        \Edge(1)(2)
        \Edge(2)(3)
        \Edge[lw=2pt](3)(4)
        \Edge[lw=2pt](4)(5)
        \Edge(5)(6)

        \tikzset{EdgeStyle/.append style = {bend left = 70}}
        \Edge[label=$<$](2)(4)

    \end{tikzpicture}
    \caption{A non-fully-routed demand disallowed by \textbf{PB}}
\end{figure}
 
\begin{lemma}
\label{lem:pseudo}
In this lemma we assume only that $f$ is a nested flow on the path $G$.
The left (resp. right) anchor sets of  a tight node $v \in S$ is a pre-certificate  as long as the following conditions hold:
\begin{enumerate}
    \item[\textbf{PA)}] $L(v) \subseteq S$ (resp. $R(v) \subseteq S$)
    \item[\textbf{PB})] there is no non-fully-routed commodity $ab$ such that $b \in L(v)$ and $a \in [i,L')$  -- see Figure~\ref{fig:PB} (resp. $a \in R(v)$ and $b \in (R',j]$).

\end{enumerate}

Here $L'$ is the leftmost anchor of $v$ (resp. $R'$ is the rightmost anchor of $v$).
\end{lemma}
\begin{proof}{\parindent0pt
It is sufficient to show that if the anchor sets $Y=\{v\}, W=L(v) \cup Y$  satisfy  properties \textbf{PA} and \textbf{PB}, then they satisfy  all properties from the Pre-Certificate Lemma \ref{lem:pseudo1}. First, note that by  \textbf{PA} we have $W   \subseteq S$ as required. \\ 

\begin{enumerate}
    \item[\textbf{P1})] Observe that $W$ is a subset of $[L',v]$ where $L'$ is the leftmost anchor. 
    Any commodity with both endpoints in $W$ is thus nested under a positive commodity that terminates at $L'$ and some node to the right of $v$.   Hence by Lemma~\ref{lem:nested-down} it is fully-routed. 
    
    \throttle{this appears to be an issue. We might route over top of it to achieve $\tau$ for some terminal no? }
    
    \item[\textbf{P2})]  For any positive commodity that transits $v$, its left endpoint  is an anchor of $v$, and   $L(v) \subseteq W$.  

    \item[\textbf{P3})] Suppose that $ab$ is a non-fully-routed commodity with an endpoint in $W$ but does not transit or terminate at $v$. By Lemma~\ref{lem:nested-down}, we must have $a < L'$ or else it is fully-routed. But  this possibility is disallowed by \textbf{PB}. 

    \item[\textbf{P4})] This follows trivially as $|Y|=1$. 
\end{enumerate}

}\end{proof}

\subsection{Proof of Theorem~\ref{thm:main}}

Recall by Assumption~\ref{ass:tight}, $S$ contains at least one tight node.

\begin{lemma}
\label{lem:left-most-tight}
If $v$ is a leftmost tight node in $S$, then there does not exist any non-fully-routed commodity $ab$ such that $[a,b] \subseteq [i,v)$ (resp. a similar statement holds for rightmost).
\end{lemma}
\begin{proof}
Suppose such a  commodity $ab$ exists. Then $v$ is a left most tight node so there are no tight nodes in $[i,v)$. No player could have blocked $ab$ so it was not maximally routed, a contradiction. 
\end{proof}

\throttle{This seems OK.}

Let $X \subseteq V(G)=[n]$. We say $X$ has been {\em processed}
by the algorithm if each element of $X$ has been incorporated. 
The {\em time} of processing $X$ refers to the first time in the execution after which all elements are processed.

\begin{lemma}
\label{lem:path-timing-cert}
    Let $\hat{f}$ be the current flow at the  time of processing $\{i,j\}$.     Then $[i,j]$ is certifiable under the flow $\hat{f}$.
\end{lemma}
\begin{proof}
    Without loss of generality, let $i$ be the last of $i,j$ to be incorporated. By Assumption \ref{ass:tight} there must be at least one tight node in $[i,j]$ or else it is certifiable.  So  let $v$ be a leftmost tight node.\\
    
    \noindent
    We now show that the left anchor sets of $v$ form a pre-certificate using Lemma \ref{lem:pseudo}.
    
    \begin{enumerate}
        \item[\textbf{PA})] Since player $i$ has just been incorporated,  no commodity with an endpoint farther left than $i$ has been routed. Hence   $L(v) \subset [i,j]$.
        \item[\textbf{PB})] As $v$ is a leftmost tight node we can apply Lemma \ref{lem:left-most-tight}. 
    \end{enumerate}

\noindent    
Hence $[i,j]$ is indeed certifiable by Lemma~\ref{lem:pseudo}.
\end{proof}

\throttle{DOnt see that the above works yet. The Problem is PA.}

\begin{lemma}
\label{lem:path-bow}
If $f$ is a flow returned by the algorithm {\sc incorporate}, then $\util(f)$ is in the core. 
\end{lemma}
\begin{proof}
    The last pair to be processed by the algorithm is $\{1,n\}$  and all pairs  have been considered prior. By Lemma \ref{lem:path-timing-cert} after a pair $\{k,\ell\}$ is processed, the corresponding subset $[k,\ell]$ was certifiable for the current flow at that time. By Lemma \ref{lem:up-closed} after adding more flow each set remains certifiable. Thus every coalition of contiguous players is certifiable for $f$, and $\util(f)$ must be in the core. 
\end{proof}

This completes the proof of Theorem~\ref{thm:main}.
We now extend these ideas to handle spiders.

\begin{repcorollary}{cor:spiders}
Let $\mathcal{T}_c$ be the class of trees (so-called spiders) with at most $1$ node of degree greater than two.
There is a polytime algorithm for producing core elements for any  multi-flow game associated with a tree in $\mathcal{T}_c$. 
\end{repcorollary}

A {\em spider} is a connected graph where every node except for one has degree at most two, we call this distinguished node the root of the spider, and the paths connected to it the legs of the spider. We will consider orienting the spider with the root on the left, and each leg leading to the right so as to easily allow us to re-use the terminology from paths. We now adapt our methods on the path to find core elements when the supply graph is a spider. We consider a modified instance where all demands between legs of the spider are removed. That is, any commodity of form $uv$ where neither $u$ or $v$ is the root and $u$ and $v$ are on different legs. In particular each such commodity transits the root. We then run {\sc incorporate} starting at the root. 

\begin{lemma}
If {\sc Incorporate} starts from the root on the modified spider, then it returns a core element for the original game. 
\end{lemma}
\begin{proof}
Consider a (contiguous) coalition $S$ that does not contain the root $r$, then the flow returned is the same as if we simply ran {\sc  incorporate} on the leg that contains $S$ starting at the root, and by our results on the path $S$ is certifiable. 

Consider a coalition $S$ that does contain the root. We can then consider each leg $\ell_i$ of the spider individually. Let $t_i$ be the rightmost node of $\ell_i$ contained in $S$, then similarly to Lemma \ref{lem:path-timing-cert} we look at the flow $\hat{f}$  at the time of processing $\{r,t_i\}$. If there are any nodes in $\ell_i$ that are tight under $\hat{f}$ then there is a rightmost tight node and a right anchor certificate. If no leg has a right anchor certificate, then there are no tight nodes to block the flow from $r$ and $r$ must be $S$-content. Hence by Lemma \ref{lem:s-cont} $S$ is certifiable.
\end{proof}

After we find a flow in the core for the modified instance we can add back the removed commodities and route them arbitrarily. 

Note we can view any path as a spider with two (potentially empty) legs with any node playing the role of the root. This gives us access to more valid orderings and potentially more core elements.  

\section{The Empirical Core}
\label{sec:empirical}

We now have an algorithm that can potentially create a different core element for each  order in which players are incorporated. In this we investigate the diversity of core elements generated  by running {\sc incorporate} multiple times with a random incorporation ordering. We refer to the set of core elements obtained as the {\em empirical core} or $ECore$ of a given game. 

Since games (both cooperative and non-cooperative) may have many stable solutions,   it is interesting and important to compare relative merits associated with them.  It is natural to view this as a multi-objective decision process where the two most common {\em social choice} metrics are {\em social welfare} (SW) and {\em fairness}. The social welfare of a payoff is the sum individual payoffs  $SW(\util)=\sum_{v \in V} \util_v$, and the fairness  is the minimum payoff of any player $Fair(\util) = \min_{v \in V} \util_v$. Maximizing social welfare is desirable because it maximizes efficiency, whereas maximizing fairness  appeals to notions of distributive justice (sometimes know as equity). Since all core elements are stable payoffs,  our main question can be summarized as "{\bf What is the impact of stability on efficiency and fairness?}"

We next describe the different classes of games whose empirical cores we investigate.

\subsection{Game Models}

Recall that an instance of a multicommodity flow game on the path is defined by the number of players, their capacities, and the demand graph with its edge weights. Each family of games is  defined by a 
 {\em game model} denoted   $M(n,C,D)$ where $n,C,D$ are three real-valued parameters; $n$ denotes the number of players
 and $C,D$ are values which govern the choice of capacity and demand vectors. Each model comes with a method or experiment for generating
 capacities and demands based on the parameters $C,D,n$.  Intuitively,
 $C,D$ are used to define how {\em competitive} the games are.   When capacities are high and demands relatively low, there is little  competition and players can route more or less freely. Whereas when capacities are low and demands are high, we say  there is high competition on the network. 

\subsubsection{Constant Model}

The constant model is extremely simple. This model deterministically returns a game where every player $v$ has capacity $c_v=C$, and every pair of players $uv$ has demand $d_{uv}=D$. 

\subsubsection{Gaussian Marginal Model}

In this model we generate marginal demands $d_v=\sum_{u \neq v} d_{uv}$ for each player $v$ according to a distribution described below. Each player's capacity is then set to be $c_v=Cd_v$. To generate the demands we draw $D$ pairs $u,v$ using a truncated normal distribution (for a formal definition see \ref{app:trunc-dist})

For each pair drawn from the distribution we round them both to the nearest integers $u,v$. For each such pair $u,v$ we increment $d_{uv}$. 

\subsubsection{Random Graph Model}

In this model each capacity is drawn uniformly at random from the interval $[1,C]$. The demand graph is generated from an Erd\"os-Renyi random graph model where every edge is included independently with probability a half. Each edge weight $d_{uv}$ is then drawn uniformly from the interval $[1,D]$.

\subsection{Results}

\subsubsection{Constant Model Results}

In this section we fix the number of players $n=50$ and the constant demand $D=1$. We then generate multiple games as we let the amount of capacity rise. Eventually as capacity is abundant enough all demands can be routed and the core consists of a single payoff. Clearly, when $C=0$ and there is no capacity the only feasible flow is the zero flow, and consequently the core must also be this single point.

We first examine how the number of distinct core elements (size of the Ecore) changes as the  capacity increases (competitiveness decreases). These results are show in Figure \ref{fig:const_count}. There are $2 \sum_{i=1}^{\frac{n-1}{2}}\binom{n-(i+2)}{i-1}$ different incorporation orderings, or for $n=50$ there are 9,615,053,952 different orderings. We however only generate a tiny fraction of this amount. 

\begin{figure}[H]
\centering
    \includegraphics[scale=0.5]{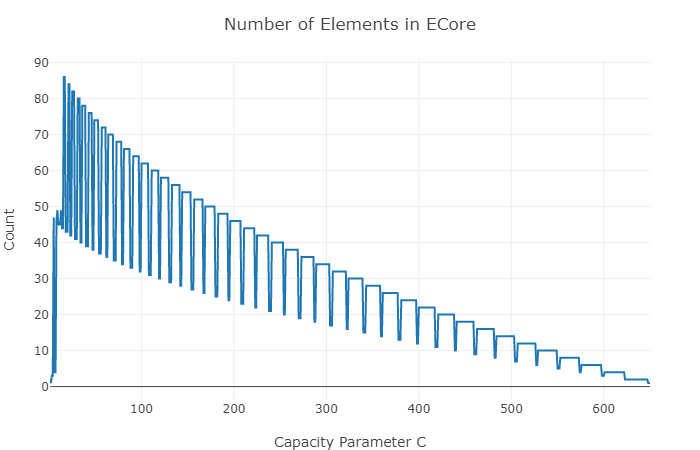}
    \caption{Number of (distinct) Elements in ECore for Constant Model, 2000 samples per game}
    \label{fig:const_count}
\end{figure}

Next, we examine the same games from the perspective of social welfare.  Each game is determined by the parameter $C$ and for each game we consider $2000$ random incorporation orderings. 
In Figure \ref{fig:const_SW} we plot the minimum, average, and maximum social welfare over the empirical core for each $C$. We also compare these to the maximum possible SW (denoted Optimal SW) which is obtained by solving a  multicommodity flow problem. Note that  the curves for the maximum SW over the empirical core is occluded by the curve  Optimal SW.

The take-away message from Figure~\ref{fig:const_SW} appears to be that for the constant model,  core elements essentially maximize social welfare. Even our minimum sampled  core elements  achieves ~93\% of the maximum social welfare. Conversely, it is natural to ask if solving the Optimal SW LP  always produces  a core element. This is false as the example of 
Figure \ref{fig:nest-counter} is a basic optimal solution to the SW LP, but it is not in the core.   

We next examine how effective the empirical core is in terms of  fairness.

\begin{figure}[H]
\centering
    \includegraphics[scale=0.5]{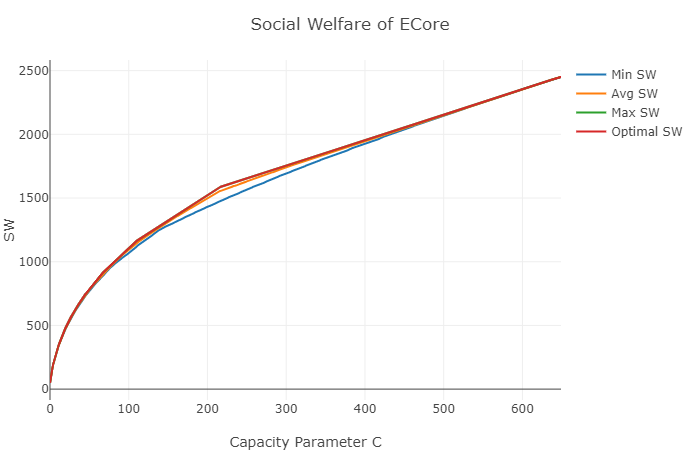}
    \caption{Social Welfare of ECore for Constant Model, 2000 samples per game}
    \label{fig:const_SW}
\end{figure}

Figure \ref{fig:const_fair} shows the minimum, average, and maximum fairness over the same family of empirical cores. Since optimal fairness (not necessarily for core payouts) can again be computed by an LP, we include this in our comparison. The message here is quite different. First, even the maximum ECore fairness is generally quite far (factor $2$) from the optimal fairness. Second, the ratio between the maximum ECore fairness to the average and minimum ECore fairness is substantial.

\begin{figure}[H]
\centering
    \includegraphics[scale=0.5]{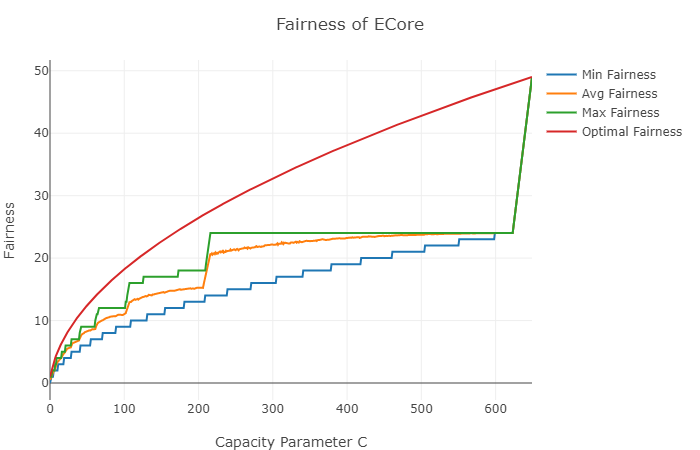}
    \caption{Fairness of ECore for Constant Model, 2000 samples per game}
    \label{fig:const_fair}
\end{figure}

\subsubsection{Gaussian Marginals Model Results}

We repeat the above experiments for the Gaussian Marginal Model in Figures \ref{fig:norm_count}, \ref{fig:norm_SW}, and \ref{fig:norm_fair}. Again in Figure \ref{fig:norm_SW} the maximum core social welfare is the same as the optimal social welfare. 

\begin{figure}[H]
\centering
    \includegraphics[scale=0.5]{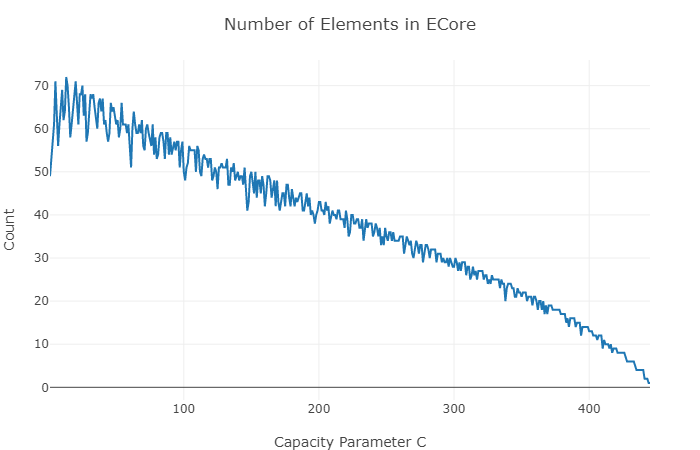}
    \caption{Number of Elements in ECore for Gaussian Marginals Model, 2000 samples per game}
    \label{fig:norm_count}
\end{figure}

\begin{figure}[H]
\centering
    \includegraphics[scale=0.5]{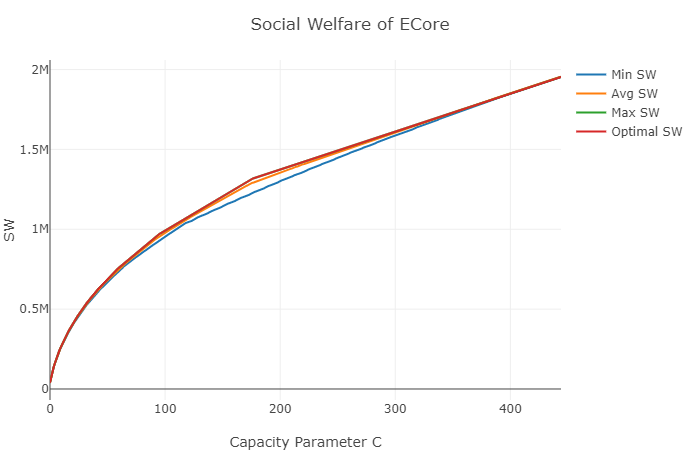}
    \caption{Social Welfare of ECore for Gaussian Marginals Model, 2000 samples per game}
    \label{fig:norm_SW}
\end{figure}

\begin{figure}[H]
\centering
    \includegraphics[scale=0.5]{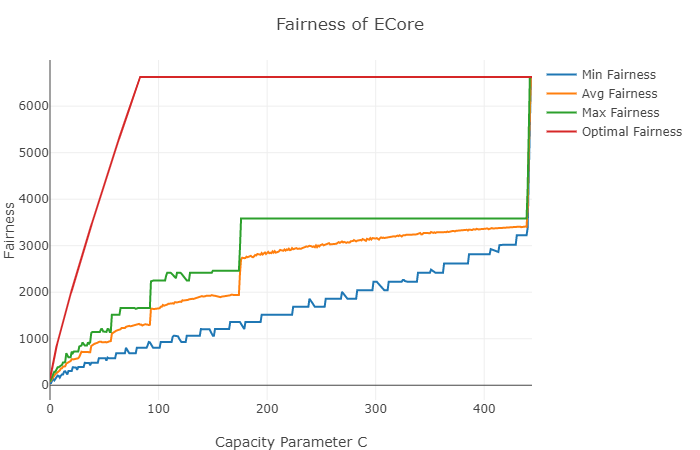}
    \caption{Fairness of ECore for Gaussian Marginals Model, 2000 samples per game}
    \label{fig:norm_fair}
\end{figure}

\subsubsection{Effect of Incorporation Time on Payoff}

In this section we investigate whether there is a benefit for individual players based on the time at which they are incorporated. It is natural to suppose that members would either wish to join the coalition either strictly earlier or strictly later. As we will see this is not the case. Given that a player's payoff depends not only on the time they were added, but also the order in which the other players are added, we investigate all players' average payoffs over the different incorporation orderings.

In Figure \ref{fig:const_time} we see this plotted for a constant model game ($n=50$, $C=250$). The most notable feature is that nearly all of the mass is concentrated to the two main diagonals. In other words, it seems on average that if a player is located at position $i$, they do best if they are added near time $i$ or time $n-i$, with a bias towards the later time. This characteristic "X" shape is present in all the models we tested including the random graph model. While we can observe "spikes" along the main diagonals, these are an artifact of the capacity parameter and not intrinsic to the players' position/their time of incorporation.

\begin{figure}[H]
\centering
    \includegraphics[scale=0.5]{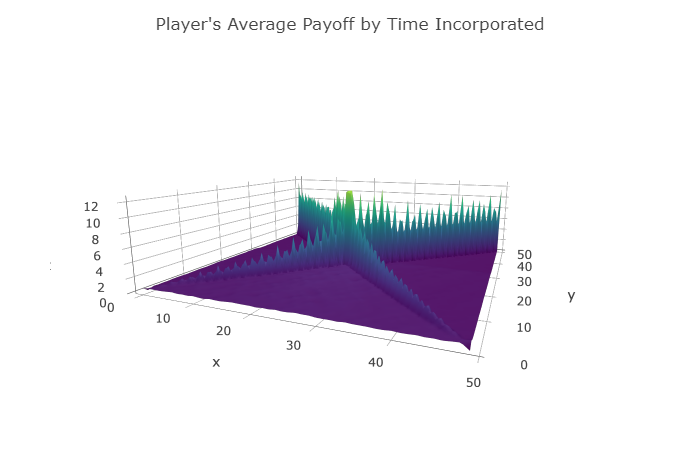}
    \caption{Player's Average Payoff by Time Incorporated, 5000 samples, x=players position, y=time incorporated, z=average payoff}
    \label{fig:const_time}
\end{figure}

\section{Fairness versus Social Welfare}
\label{sec:fairness}

In this section we examine the fairness-social welfare tradeoff in more detail. We start by showing that for single-sink instances, we can balance these performance objectives perfectly.

\begin{reptheorem}{thm:SS}
Let $\mathcal{F}$ be the family of multiflow instances $(G,H)$ where
$G$ is any capacitated supply graph  and $H$ is a single-sink commodity graph. Then there is a polytime algorithm which given an instance in $\mathcal{F}$, produces a core vector which simultaneously achieves the maximum social welfare
and maximum fairness.
\end{reptheorem}
\begin{proof}
We call $x \in \mathbb{R}_{\geq 0}^{[k]}$ a {\em routable vector} if there is a feasible flow which routes $x_i$ units from each terminal $s_i$.
As observed in Section~\ref{sec:emp} a flow induces a core element if   it maximizes $\sum x_i$.  
We next  observe that
$P=\{x \in \mathbb{R}^{[k]}_{\geq 0}: \mbox{$x$ is a routable vector}\}$  is  a polymatroid (cf. \cite{schrijver2003combinatorial}).
As such the greedy algorithm always produces a maximum flow. In other words, we may  greedily process the terminals $s_i$ in any order.  At each step, we route  additional flow 
from $s_i$ without reducing flow from  previously processed terminals (using a standard Ford-Fulkerson Flow algorithm, this may require  re-routing paths). We may perform this process so that terminals are visited more than once, if in earlier iterations we do not ``max out'' the flow from a terminal.

Obviously this algorithm could lead to a maximum flow where many terminals route $0$ ($x_i=0$). 
Hence (as with our experiments on the line) many of the maximum flows $x \in P$  will be core elements which perform poorly in terms of fairness.  
This can be fixed  by adding a  parameter $\tau$ which is a fairness target;  we can later do binary search on $\tau$.
We run a first {\em throttled phase} of the greedy algorithm where we do not increase any $x_i$ above $\tau$. If this results in a flow vector with some $x_i < \tau$, then
we selected $\tau$ too large for fairness.  Otherwise, we perform  a second greedy pass of the terminals. This increases some of the flows
and  is guaranteed  to produce a maximum flow (by polymatroidality) and hence a core vector with fairness $\geq \tau$.
\end{proof}

The proof suggests  a natural way to induce fairness out of any   core-producing algorithm.  Run a throttle phase to guarantee fairness, and then run a second phase which produces a core vector with large social welfare.  Unfortunately, the $P_4$ example (Section~\ref{sec:empirical}) shows that this fails even for the multiflow game.
Instead we relax the target of producing a core vector.
We only require a vector in the approximate core, or equivalently,  a core vector in the game where capacities are scaled down uniformly.

In the following we let $LP_{sw}$ and $LP_{fair}$ denote the LPs which maximize social welfare/fairness for the set of feasible payoff vectors in a given game.

.
\begin{reptheorem}{thm:bicriteria}
Consider a core-producing algorithm for a family of downwards closed games. Consider an instance for which $LP{fair} \geq \tau$.
For any $\lambda \in (0,1)$ we can compute a $\frac{1}{1-\lambda}$-approximate core vector with fairness at least $\lambda \tau$.  
\end{reptheorem}
\begin{proof}
We run the given algorithm for the instance scaled down by 
$(1-\lambda)$. It follows from Lemma~\ref{lem:scale} that the output $\util^{(1)}$ is in the $\frac{1}{1-\lambda}$-approximate core. 
By assumption, there exists $\util^{(2)} \in \Util(N)$ such that $Fair(\util^{(2)}) \geq {\bf \tau}$. Thus (by downwards-closed) $\util=\util^{(1)}+\lambda \util^{(2)} \in \Util(N)$.  We also have that $\util \geq \util^{(2)}$ and hence the same reasoning used in Lemma~\ref{lem:up-closed} implies that $\util$ is
again in the approximate core.
\end{proof}

\section{Conclusion}
\label{sec:conc}

We have provided an efficient algorithm for generating multiple core elements when the transit network is a path or spider. 
 We have shown while there appears to be little trade-off between efficiency and stability, there is a trade-off between stability and fairness.  Our theoretical results are based on  certifying that a coalition has no deviation in general graphs. We hope this may be useful in the challenging problem of computing core vectors for general instances. Another interesting direction, motivated by inter-domain routing behaviours, is to consider games with penalties/taxations for transiting traffic \cite{Shepherd05}.

\bibliographystyle{alpha}
\bibliography{sample}

\appendix

\section{Truncated Normal Distribution}\label{app:trunc-dist}

Let $N(\mu,\sigma)$ be a normal distribution. Let $\phi(x)=\frac{1}{\sqrt{2\pi}}\exp \big(-\frac{1}{2}x^2 \big)$, and $\Phi(x)=\frac{1}{2} \big(1+\text{erf}(x/\sqrt{2} \big)$ be the probability density function and cumulative distribution function for the standard normal distribution respectively. Then a truncated normal distribution with a lower bound of $L$ and an upper bound of $U$ has the following probability density function:

\[ f(x) = \frac{\phi \big( \frac{x-\mu}{\sigma} \big)}{\sigma\Big( \Phi \big(\frac{U-\mu}{\sigma} \big)-\Phi \big(\frac{L-\mu}{\sigma} \big) \Big)}  \]

For our model we choose $\mu=\frac{n}{2}$, $\sigma=2\sqrt{n}$, and $L=1$, $U=n$.

\end{document}